\documentclass[12pt,oneside]{article}
\usepackage{amssymb,amsmath,amsthm,latexsym}
\usepackage{graphicx}
\usepackage[utf8x]{inputenc}
\usepackage{wrapfig}
\usepackage{color}
\usepackage{titlefoot}
\usepackage{authblk}
\def\Fg {{\cal F}} 
\def\Hg {{\cal H}} 
\def\Sg {{\cal S}} 
\newtheorem{thm}{Theorem}
\newtheorem{lem}{Lemma}
\newtheorem{cor}{Corollary}
\newcommand{\red}[1]{#1}

\title{\LARGE \bf Non-relativistic Pauli--Fierz Hamiltonian for less than two photons}
\author{D. Dayantsolmon, A. Galtbayar}

\begin{document}

\maketitle

\begin{abstract}
We consider the Pauli--Fierz model, which describes a particle (an electron) coupled to the quantized electromagnetic field and limit the number of photons to less than 2. 
By computing \red{the resolvent explicitly}, we located the spectrum of the Hamiltonian   mass. 
Our results do not depend on the coupling constant $e$ nor on the infrared cutoff parameter $R$. 
\end{abstract}

\unmarkedfntext{2010 {\it Mathematics Subject Classification} 81V10, (81T16, 47A10, 47A75) \\

{\it Keywords}: Pauli--Fierz Hamiltonian, mass renormalization, dressed electron states}

\section{Introduction}\label{resseq1}
In this paper, we study the fiber Hamiltonian for the standard model of non-relativistic quantum electrodynamics, called \red{the} Pauli--Fierz model, when the number of photons is restricted to 0 or 1. The latter condition allows us to compute the resolvent of the fiber Hamiltonian explicitly. Therefore, the spectrum and the effective mass can be obtained for \red{arbitrary} values of parameters such as \red{the} coupling constant with the electromagnetic field and \red{the} ultraviolet cutoff radius.

There is a rich literature on the spectral and scattering properties of this model. The spectral properties  of the Pauli-Fierz Hamiltonian was studied in \cite{GGM} and the existence of the ground states of the Pauli-Fierz Hamiltonian was proved in \cite{GLL}. See \cite{F2} for discussions on the spectral properties and scattering theory of the Nelson Hamiltonian. The spin-boson \red{model} and the Nelson model are discussed in \cite{MS} and \cite{GJY} for when the number of photons is restricted to a few.  This work was initially inspired by the paper \cite{M}, where the author considered a model with less than two phonons {\em without polarization} and computed the spectrum, the ground state, as well as the effective mass.

An extensive review of the properties of the ground state of the fiber Hamiltonian, its differentiability, and the effective mass can be found in \cite{B}. Most of the previous results were derived for various conditions for the above-mentioned parameters and for a certain limited range of the total momentum. 
We note that for the Nelson model, the spectrum shifts from zero to negative values due to the radiation field, which is not observed in our setup.

Let us introduce the model. 
We set the bare electron mass $m$ and the speed of light $c$ to be equal to 1. 
The Hilbert space for the system is given by
$\Hg:=L^2(\mathbb{R}_x^3)\otimes \Fg$,
where the bosonic Fock space $\Fg$ is defined by
$$
\Fg=\bigoplus_{n=0}^\infty\Fg^{(n)}=\bigoplus_{n=0}^\infty\left(\otimes^n_s \Sg\right),
$$
with $\Sg=L^2(\mathbb{R}^3)\oplus L^2(\mathbb{R}^3)$. 
We have denoted the $n$-fold symmetric tensor product of $\Sg$ by $\otimes^n_s \Sg$, with $\otimes^0_s \Sg=\mathbb{C}$. 
The annihilation and the creation operators $a$ and $a^*$ are defined as
\begin{equation}\label{g1}
a^\sharp(v)=\sum_{\lambda=1}^2 \int a^\sharp(k,\lambda) v(k,\lambda)dk,
\end{equation}
for $v=(v(\cdot,1), v(\cdot,2))\in L^2(\mathbb{R}^3)\oplus L^2(\mathbb{R}^3)$, 
where $a^\sharp$ is for $a$ or $a^*$ and $a^\sharp(k,\lambda)$ is a formal kernel. 
The free photon field operator $H_f$ on $\Fg$ is defined as
$$
H_f=\sum_{\lambda=1}^2 \int\omega(k)a^\ast(k,\lambda) a(k,\lambda) dk,
$$
where $\omega(k)=|k|$ is the photon energy.
The quantized radiation field $A_{g}(x)=(A_{g1}(x),A_{g2}(x),A_{g3}(x))$, $x\in \mathbb{R}^3$, acting on $\Fg$ is given by
$$
A_{gj}(x)=\frac{1}{\sqrt2} \sum_{\lambda=1}^2 \int e_j(k,\lambda)\left[g(k)e^{-ik\cdot x} a^*(k, \lambda)+ \overline{g(k)}e^{ik\cdot x} a(k, \lambda) \right] dk.
$$
Here, $e(k,\lambda)=(e_1(k,\lambda), e_2(k,\lambda),e_3(k,\lambda))$ are polarization vectors satisfying the conditions $k\cdot e(k,\lambda)=0$ and $e(k,\lambda)\cdot e(k,\mu)=\delta_{\lambda \mu}$,  $\lambda,\mu=1,2$. \\

\paragraph{Assumption.} We assume that 
\begin{equation}\label{factor}
    g(k) = \frac{\chi(k)}{\omega^{1/2-\sigma}(k)},
\end{equation}
where $\chi(k)$ is a characteristic function of region $\{k\in\mathbb{R}^3 \mid |k|\leq R \}$, $R>0$ is \red{the} ultraviolet cutoff radius, 
and $0\leq\sigma<1/2$ is \red{an} infrared renormalization parameter.
 
Note that the infrared renormalization was introduced only to remove singularities of some auxiliary integrals that appear later. Our main results hold for all values of $\sigma$, including zero. 

Finally, the Pauli--Fierz Hamiltonian is defined as
\begin{equation}
H=\frac12 (-i\nabla_x\otimes1 -eA_g(x))^2+ 1\otimes H_f,
\end{equation}
where $e$ is the charge of the electron or the coupling constant with the field. 
For assumption (\ref{factor}), it was proved in \cite{HS2} that $H$ is self-adjoint on the domain $D(-\dfrac12\Delta+H_f)$ for arbitrary values of the coupling constant $e$. 
We define the total momentum operator on $\Hg$ as
\begin{equation}
P_\text{tot}=-i\nabla_{x}\otimes 1+1\otimes P_f, 
\end{equation}
where $P_f=\sum_{\lambda=1}^2 \int k a^\ast(k,\lambda) a(k,\lambda) dk$ is the photon momentum. 
Since $H$ commutes with $P_\text{tot}$, it can be decomposed with respect to the spectrum of $P_\text{tot}$:
$$
H=\int_{\mathbb{R}^3}^{\oplus} \bar H(p) dp,
$$
where $\bar H(p)$ is defined as
\[
\bar H(p) = \frac12 (p-P_f-eA_g(0))^2+H_f
\]
on $\Fg$. Now $p\in \mathbb{R}^3$ is considered as a parameter.
Let $E_p$ be the projection operator onto the 0 or 1 photon space. We introduce the corresponding Pauli--Fierz operator:
\begin{equation}\label{PFcut}
H(p):=E_p\bar H(p)E_p.
\end{equation}
In the following, we will work only with the operator $H(p)$. 

The rest of this paper is organized as follows. 
We state \red{our} main results in Section~\ref{resseq2}
and compute the resolvent of $H(p)$ in Section~\ref{resseq3}. 
Then in Section~\ref{resseq4}, we locate the spectrum of $H(p)$, 
and \red{finally} in Section~\ref{resseq5}, we prove our main results.

\section{Main results}\label{resseq2}

We are interested in the spectral properties of the Pauli--Fierz Hamiltonian, and in particular, in finding the resolvent of the corresponding fiber Hamiltonian and calculating the effective mass.  

The fiber Hamiltonian $H(p)$, defined  in the 0 or 1 photon space, is a finite-rank perturbation of an operator \red{whose spectrum consists} only of an absolutely continuous \red{part}.
Therefore, it is very similar to the Friedrichs Hamiltonian. 
For the spectral characterization of the Friedrichs Hamiltonian and related results, see \cite{F1} and \cite{HSp3}.

\red{Restricting the} number of photons allows us to express the resolvent explicitly. Therefore, all \red{our} results were obtained in a nonperturbative way and are independent of parameters such as $e$ \red{and} $R$. In addition, this enables us to work on the scattering properties of this model, as in \cite {GJY}, which, \red{however,} will be discussed elsewhere.

We are now ready to formulate the main results. 
For any $p\in \mathbb{R}^3$, let
\begin{equation}\label{gg108}
\begin{aligned}
z_{0}(|p|) &= \min_{k\in \mathbb{R}^3}\left\{\frac12(p-k)^2+|k|+\gamma_0\right\} \\
           &=
\begin{cases}
    \dfrac{1}{2} p^2+\gamma_0,     & \text{if } |p|\leq 1, \\
    |p| - \dfrac{1}{2} + \gamma_0, & \text{if } |p|>1, \\
\end{cases}
\end{aligned}
\end{equation}
be the curve on the $(|p|,z)$ plane, where $\gamma_0$ is \red{a} constant depending on $e$ \red{and} $R$. 
An example to keep in mind is 
\[
\gamma_0=\frac{\pi}{1+\sigma} e^2 R^{2+2\sigma}.
\]
We define the function $F(p,z)$ as 
\begin{multline}
    F(p,z)= 
    \frac12p^2-z+\gamma_0 \\
    - \pi e^2\left(\frac12p^2+z-\gamma_0\right) 
    \int_0^R \int_{-1}^1 \frac{(1-t^2)dt\rho^{1+2\sigma} d\rho}{p^2/2-|p|\rho t + \rho^2/2+\rho+\gamma_0-z},
\end{multline}
which is crucial for finding the eigenvalue of the reduced operator $H(p)$. 
\red{In Lemma~\ref{dd2} of Section~\ref{resseq4}, we will derive} the following properties \red{of} the function $F(p,z)$ in the region $\Omega^-=\{(|p|,z) \mid z \leq z_{0}(|p|), z\in \mathbb{R} \}$:
\begin{itemize}  
\item $F(p,z)$ is real analytic and \red{is a} decreasing function on $z$.
\item The equation $F(p,z)=0$ has a unique solution $z=z^\ast(p)$ when $|p|\leq 1$ and there exists a constant $p_0>1$ such that it has no solution when $|p|>p_0$.
\end{itemize}

\begin{thm}[Spectrum of $H(p)$]\label{th1} 
For any $p\in \mathbb{R}^3$, the spectrum of $H(p)$ consists of the essential spectrum $ [z_0(|p|), +\infty)$ and the eigenvalue $z^*(p)$, which is a solution of the equation $F(p,z)=0$ in the region 
$$
\Omega^-=\{(|p|,z) \mid z \leq z_{0}(|p|), z\in \mathbb{R} \}.
$$
Moreover, we have the following bound for $z^*(p)$:
$$
0<z^*(p)\leq\frac{7}{2}+\frac{ 4\pi e^2R^{1+2\sigma}}{1+2\sigma}+\gamma_0.
$$ 
\end{thm}

The eigenvalue $ z^*(p)$, when it \red{exists}, is the infimum of the spectrum $ H(p)$, which we denote by $ E_{\sigma}(p) $. 
Knowing the exact value of $ E_{\sigma}(p) $ \red{would allow us to} calculate the effective mass $ m_{\rm eff} $, 
\red{defined through}
$$
E_\sigma(|p|) -E_\sigma (0) = \frac{p^2}{2m _ {\rm eff}} + O (|p|^3)
$$
for small $p$. When $E_\sigma(|p|)$ is \red{a} $C^2$-function \red{in a neighborhood of} $p=0$, as a direct consequence of the \red{preceding} definition, we have
\begin{equation}\label{eff1}
    \frac{1}{m_{\rm eff}}=\lim_{p\to 0} \frac{\partial^2 E_{\sigma}(p)}{\partial |p|^2}.
\end{equation}

\begin{thm}[Effective mass] \label{th2}
\red{The function}
$E_\sigma(|p|)$ is \red{a} $C^2$-function \red{in a neighborhood of}  $p=0$ and the effective mass for $H(p)$ has the following form:
\begin{equation}\label{eff2} 
    \frac{1}{m_{\rm eff}} = \frac{1-\pi e^2 D_{12}(0,\gamma_0)}{1+ \pi e^2 D_{12}(0,\gamma_0)}
\end{equation}
where 
$$
D_{12}(p,z) = \int_0^R \int_{-1}^1 \frac{(1-t^2)\rho^{1+2\sigma} dtd\rho}{p^2/2-|p|\rho t + \rho^2/2+\rho+\gamma_0-z}.
$$
\end{thm}

We consider the special case as $\sigma\to 0$, which removes the infrared renormalization.
 
\begin{cor} For any values of $e$ and $R$, we have 
\begin{equation}\label{eff3} 
    \lim_{\sigma\to 0}\frac{1}{m_{\rm eff}}=\frac{1-(8/3)\pi e^2 \ln(R/2+1)}{1+ (8/3)\pi e^2 \ln(R/2+1)}.
\end{equation}
\end{cor}

Expansion of the effective mass in \red{terms} of the fine structure constant $\alpha=e^2/4\pi$ was done in \cite{HS1}, \cite{BCJI}, assuming the constant $e$ be a small. Note that our result gives the effective mass for arbitrary values of $e$ and $R$, and it is consistent with the results of \cite{HS1} and \cite{BCJI} when $e^2 \ln(R/2+1)\approx o(1)$. Indeed, from formula (\ref{eff3}), we can derive that  
$$
m_{\rm eff}\approx 1+\frac{16}3\pi e^2 \ln \left( \frac{R}{2} + 1 \right).
$$ 

\section{The resolvent of $H(p)$} \label{resseq3}

Before \red{proving} the main results, we will calculate the resolvent of the operator ${H}(p)$ defined by (\ref{PFcut}), in the space $\mathcal{H}=\mathbb{C}\oplus L^2(\mathbb{R}^3)\oplus L^2(\mathbb{R}^3)$. 

For any $p\in\mathbb{R}^3$ and $f=(f_0(p), f_1(p,k,1), f_1(p,k,2))^t\in \mathcal{H}$, the matrix form of $H(p)f$ is 
\begin{equation}\label{H(p)}{\small{
\hspace{-0.5cm}\left(\begin{array}{ccc}
\widetilde{T}(p)  &  -\frac{e}{\sqrt2}p\cdot \langle G(1)| & -\frac{e}{\sqrt2}p\cdot \langle G(2)| \\
-\frac{e}{\sqrt2}p\cdot | G(1)\rangle &  \widetilde{L}(p,k)+\frac{e^2}{2} |G(1)\rangle \cdot\langle G(1)| & \frac{e^2}{2} |G(1)\rangle\cdot \langle G(2)| \\
-\frac{e}{\sqrt2}p\cdot | G(2)\rangle  & \frac{e^2}{2} |G(2)\rangle\cdot \langle G(1)| &
\widetilde{L}(p,k)+\frac{e^2}{2} |G(2)\rangle\cdot \langle G(2)|
\end{array} \right)\left(
\begin{matrix}
f_0(p)\\
f_1(p,k,1)\\
f_1(p,k,2)
\end{matrix}\right) }},
\end{equation} 
where the annihilation and creation operators in $\mathcal{H}$ are denoted \red{by}
$$ 
\langle G(\lambda)| v = \int \overline{G(k,\lambda)}v(k) dk \quad \text{and} \quad | G(\lambda)\rangle v=G(k,\lambda)v(k),
$$
respectively, for each polarization direction $\lambda=1,2$.  
\red{Here, we have introduced the notation}
$G(k,\lambda)=e(k,\lambda)g(k)$, and \red{note that}
$$
\| G\|^2=\sum_{\lambda=1}^{2}\int|e(k,\lambda)|^2g^2(k)dk=\frac{4\pi}{1+\sigma} R^{2+2\sigma}. 
$$
The elements on the diagonal are 
\begin{align}
\widetilde{T}(p)   &= \frac12 p^2+ \frac{e^2}{4}\|G\|^2 \\
\widetilde{L}(p,k) &= \frac12(p-k)^2+\omega(k)+\frac{e^2}{4}\| G\|^2.
\end{align}
To find the resolvent of $H(p)$, we need to solve the equation 
\begin{equation}\label{reseq19}
(H(p)-z)f=u
\end{equation}
for a given $u=(u_0(p), u_1(p,k,1), u_1(p,k,2))^t$. 

For ease of writing, we also use the following notation:
\begin{align*}
T              &= \widetilde{T}(p)-z, \\
L              &= \widetilde{L}(p,k)-z, \\
b_\lambda(p,k) &= -\frac {e}{\sqrt 2} p\cdot G(k,\lambda), & \lambda=1,2, \\
N(p,k,\lambda) &= b_\lambda(p,k)p+\frac{Te}{\sqrt2}G(k,\lambda), & \lambda=1,2.
\end{align*}

\begin{lem} \label{lem11}
For any $z\in \mathbb{C}\setminus \mathbb{R}$, the solution of Equation~(\ref{reseq19}) can be written as 
\begin{align}
f_0(p)           &= \frac{1}{T}\left[u_0(p)+p\cdot (S(p,z)U^{-1})\right]  \label{g1.16}\\
f_1(p,k,\lambda) &= \frac{1}{TL}\left[Tu_1(p,k,\lambda)- b_\lambda(p,k)u_0(p)- N(p,k,\lambda)\cdot (S(p,z)U^{-1})\right],\label{g1.17}
\end{align}
where  $\lambda=1,2$, 
$$S(p,z)=\frac{e}{\sqrt 2}\int \sum_{\lambda=1}^2 \overline{G(k,\lambda)}\left(\frac{1}{L}u_1(p,k,\lambda)-\frac{1}{TL} b_\lambda(p,k)u_0(p)\right)dk,
 $$
and the matrix $U:=(u_{ij})_{i,j=1}^3$ is given by
\begin{equation}\label{g103}
   u_{ij}=
\begin{cases}
a+b p_i^2, & i=j,  \\
bp_ip_j ,  & i\neq j,
\end{cases}
\end{equation}
with $a=\dfrac{2+D_1+D_2}{2}$ and $b= \dfrac{D_1-3D_2}{2p^2}-\dfrac{D_1-D_2}{T}.$
Here,
\begin{align}
\label{04.091}
   D_1(p,z) &= \pi e^2 \int_0^R \int_{-1}^1 \frac{\rho^{1+2\sigma} dt d\rho}{p^2/2-|p|\rho t + \rho^2/2+\rho+\gamma_0-z}, \\
\label{04.092}
D_2(p,z) &= \pi e^2 \int_0^R \int_{-1}^1 \frac{t^2\rho^{1+2\sigma} dtd\rho}{p^2/2-|p|\rho t + \rho^2/2+\rho+\gamma_0-z}.
\end{align}
\end{lem}

\begin{proof} 
By introducing the notation
\begin{equation}
Q = \sum_{\lambda=1,2}\frac{e}{\sqrt 2}\int \overline{G(k',\lambda)}f_1(p,k',\lambda)dk', \label{q1}
\end{equation} 
Equation~(\ref{reseq19}) can be written using the matrix form  (\ref{H(p)}) of $H(p)$, as
\begin{align}\label{04.099}
Tf_0(p)-p\cdot Q &= u_0(p) \notag\\ 
b_\lambda(p,k)f_0(p)+Lf_1(p,k,\lambda)+\frac{e}{\sqrt2}G(k,\lambda)\cdot Q &= u_1(p,k,\lambda),\quad \lambda=1,2.
\end{align}
\red{Upon} solving (\ref{04.099}), we obtain  
\begin{align}
f_0(p) &= \frac{1}{T}\left[u_0(p)+p\cdot Q \right], \label{g6a} \\
f_1(p,k,\lambda) &= \frac{1}{TL}\left[Tu_1(p,k,\lambda)- b_\lambda(p,k)u_0(p)- N(p,k,\lambda)\cdot Q\right]. \label{g4a}
\end{align}
To conclude the proof, it suffices to find $Q$. We substitute (\ref{g6a}) \red{and} (\ref{g4a}) into (\ref{q1}) to get an equation for $Q$:
\begin{equation}\label{g10}
Q+\frac{e}{T\sqrt 2}\sum_{\lambda=1,2}\int \frac{1}{L} \overline{G(k,\lambda)}N(p,k,\lambda)\cdot Qdk =
S(p,z). 
\end{equation}
For $k\neq 0$, let $\hat{k}:=k/|k|$. 
Using the identity 
$$
p-(\hat k,p)\hat k=(e(k,1)\cdot p)e(k,1)+( e(k,2)\cdot p)e(k,2),
$$ 
Equation~(\ref{g10}) can be written as
\begin{multline}
Q
-\left(\frac{e^2}{2T}\int \frac{g^2(k)}{L} \left(p-(p\cdot \hat k) \hat k\right)dk\right)(p\cdot Q) \\
+\left(\frac{e^2}{2}\int \frac{g^2(k)}{L} \left(Q-(Q\cdot \hat k) \hat k\right)dk\right)=S(p,z). \label{q2}
\end{multline} 
Next, we show that all the integrals in (\ref{q2}) can be reduced to certain combinations of the integrals $D_1$ and $D_2$, which are defined in (\ref{04.091}) and (\ref{04.092}). 
It is easy to derive that
\begin{equation}\label{g101}
\begin{aligned}
\dfrac{e^2}{2}\int \frac{g^2(k)}{L} dk &= \pi e^2 \int_0^R \int_{-1}^1 \frac{\rho^{1+2\sigma} dt d\rho}{p^2/2-|p|\rho t + \rho^2/2+\rho+\gamma_0-z} \\
  &= D_1(p,z),
\end{aligned}
\end{equation}
and hence, we can rewrite Equation~(\ref{q2}) in a simple matrix form as
\begin{equation}\label{q2.5}
Q\left(E+\dfrac{1}{T}\left( D_1 E-C\right)\left( T E-p^tp\right)\right)=S(p,z).
\end{equation} 
Here, $E$ is a $3\times 3$ unit matrix and  
\begin{equation} \label{C}
C=\left(\dfrac{e^2}{2}\int \dfrac{g^2(k)k_ik_j}{Lk^2}dk\right)_{i,j=1}^3.
\end{equation}  
To calculate the elements of \red{the} matrix $C$, we introduce the following spherical coordinate system $(\rho, \varphi, \theta)$ where $0\leq \varphi<2\pi$ and $0\leq \theta<\pi$. We take the zenith direction \red{to be} $\vec l_1(p)=\hat{p}$, 
and the azimuth direction \red{to be} an orthogonal vector $\vec l_2(p)=(p_2,-p_1,0)/p_+$. Here, $p_+=\sqrt{p_1^2+p_2^2}$ and $\hat{p}=p/|p|$. Then, any vector $k=(k_1, k_2, k_3)$ can be written as 
$$
k=a_1\vec l_1(p)+a_2\vec l_2(p)+a_3\vec l_3(p),
$$
where \red{the} third orthogonal vector is $\vec l_3(p)=(p_1p_3,p_2p_3,-p_1^2-p_2^2)/(|p|p_+)$ and 
\begin{align*}
a_1 &= (k,\vec l_1(p))=\rho\,\cos\theta, \\
a_2 &= (k,\vec l_2(p))=\rho\,\cos\varphi\,\sin\theta, \\
a_3 &= (k,\vec l_3(p))=\rho\,\sin\varphi\,\sin\theta.
\end{align*}
This \red{gives us}
\begin{equation}\label{Sphere}
k^t=\begin{pmatrix}
k_1 \\[4mm]
k_2\\[4mm]
k_3 
\end{pmatrix}
=
\begin{pmatrix}
\dfrac{p_1}{|p|}\rho\,\cos\theta+ \dfrac{p_2}{p_+} \rho\,\cos\varphi\,\sin\theta+\dfrac{p_1p_3}{|p|p_+} \rho\,\sin\varphi\,\sin\theta \\[4mm]
\dfrac{p_2}{|p|}\rho\,\cos\theta- \dfrac{p_1}{p_+} \rho\,\cos\varphi\,\sin\theta+\dfrac{p_2p_3}{|p|p_+} \rho\,\sin\varphi\,\sin\theta  \\[4mm]
\dfrac{p_3}{|p|}\rho\,\cos\theta-\dfrac{p_+}{|p|} \rho\,\sin\varphi\,\sin\theta \\ 
\end{pmatrix} .
\end{equation}
Now, computing the elements of $C$, in the aforementioned basis, we get
\begin{equation}\label{g102}
   \frac{e^2}{2}\int \frac{g^2(k)k_ik_j}{L|k|^2} dk =
\begin{cases} 
\dfrac{D_1(p^2-p_i^2) + D_2(3p_i^2 -p^2)}{2p^2}, & \text{if } i=j, \\[3ex]
-\dfrac{p_ip_j(D_1 - 3D_2)}{2p^2},               & \text{if } i\neq j,
\end{cases}
\end{equation}
where $D_1$ \red{and} $D_2$ are defined in (\ref{04.091}) and (\ref{04.092}). 
As an example, let us compute one of the elements of the matrix $C$:
\begin{align*}
c_{23} = c_{32} &= \frac{e^2}{2}\int \dfrac{g^2(k)k_2k_3}{Lk^2}dk\\
       &= \pi e^2\int_0^{R} d\rho \int_0^{\pi}  \dfrac{\rho^{1+2\sigma}}{L}\left(\dfrac{p_2p_3}{p^2}\cos^2\theta- \dfrac{p_2p_3}{2p^2}(1-\cos^2\theta)\right)\sin\theta d \theta\\
       &= -\dfrac{p_2p_3}{2p^2}(D_1-3D_2).
\end{align*} 
Now invoking the identities  (\ref{g101}) \red{and}  (\ref{g102}) in (\ref{q2.5}), 
one can obtain the equation $QU=S(p,z)$, with $U$ defined as in (\ref{g103}). 
Then, substituting $Q=S(p,z)U^{-1}$ into (\ref{g6a}) \red{and} (\ref{g4a}) \red{finally establishes the proof}.
\end{proof}

\section{The spectrum of $H(p)$}\label{resseq4}

In this section, we describe the spectrum of $H(p)$ for each $p\in \mathbb{R}^3$.
We \red{make the decomposition} $H(p)=H_0(p)+W(p)$,
where
\begin{equation}\label{dd3}
   H_0(p)=
 \begin{pmatrix}
\widetilde{T}(p)  &  0 & 0 \\
0 &  \widetilde{L}(p,k) & 0 \\
0  & 0 &
\widetilde{L}(p,k)
\end{pmatrix} 
\end{equation} 
and
\begin{equation}\label{W(p)}
W(p)=\dfrac{e}{\sqrt2}
 \begin{pmatrix}
0  & -p\cdot \langle G(1)| & -p\cdot \langle G(2)| \\
-p\cdot | G(1)\rangle &  \dfrac{e}{\sqrt2} |G(1)\rangle \cdot\langle G(1)| & \dfrac{e}{\sqrt2} |G(1)\rangle \cdot\langle G(2)| \\
-p\cdot | G(2)\rangle  & \dfrac{e}{\sqrt2} |G(2)\rangle\cdot \langle G(1)| & \dfrac{e}{\sqrt2} |G(2)\rangle\cdot \langle G(2)|
\end{pmatrix}.
\end{equation}
Note that the spectrum of $H_0(p)$ consists only of the essential spectrum, which is $[z_0(|p|),+\infty)$. 
Since $H(p)$ is \red{a} finite rank perturbation of $H_0(p)$, by Weyl's theorem, the essential spectrum of the operator $H(p)$ remains the same. 
From (\ref{g1.16}) \red{and} (\ref{g1.17}) of Lemma~\ref{lem11}, one can \red{then} see that the only possible addition to the spectrum in the interval $(-\infty, z_0(|p|))$ could be the zeros of the function $\det U$. 

The remainder of this section will be devoted to finding these zeros for each $p\in \mathbb{R}^3$. From (\ref{g103}), we derive that\begin{equation} \label{KKK} 
\begin{aligned}
K(p,z) &:= \det U= a^3+a^2bp^2 \\
       &= \dfrac{(D_1 + D_2 + 2)^2}{4T}(T - (D_1-D_2)p^2 + (D_1 - D_2)T) \\
       &= \dfrac{(D_1 + D_2 + 2)^2}{4T}\left[\frac12p^2-z+\gamma_0 - (D_1 - D_2)\left(\frac12p^2+z-\gamma_0\right)\right].
\end{aligned}
\end{equation}
Since $(D_1 + D_2 + 2)^2>0$, it is important to know the behavior of 
$$
D_{12}:=\frac{1}{\pi e^2}(D_1-D_2) 
$$
when finding the zeros of $K(p,z)$. The next lemma gives an estimate for the function $D_{12}$ on the curve $z_0(|p|)$.

\begin{lem}
\label{dd1}
On the curve $z= z_0(|p|)$, we have the following estimate for $D_{12}$:
\begin{equation}\label{eqz8}
  D_{12}(p,z_0(|p|))\leq 
  \begin{cases}
  \dfrac{4}{3(1+2\sigma)}\cdot 
  \dfrac{R^{1+2\sigma}}{1-|p|},                          & \text{if } |p| < \dfrac{1}{2}, \\[3ex]
  \dfrac{4 }{1+2\sigma}\cdot \dfrac{R^{1+2\sigma}}{|p|}, & \text{if } |p| \geq \dfrac{1}{2}.
\end{cases}
\end{equation}
\end{lem}

\begin{proof} 
From (\ref{04.091}) \red{and} (\ref{04.092}), we \red{infer}
$$
D_{12}(p,z)= \int_0^R \int_{-1}^1 \frac{(1-t^2)\rho^{1+2\sigma} dtd\rho}{p^2/2-|p|\rho t + \rho^2/2+\rho+\gamma_0-z}.
$$
When $|p| < 1/2$, we have $z_0(|p|) = p^2/2+\gamma_0$ and  
\begin{align*}
D_{12}\left(p,\frac{p^2}{2}+\gamma_0\right) &\leq \int_0^R \int_{-1}^1 \frac{(1-t^2)\rho^{2\sigma} dtd\rho}{1-|p|} \\
&= \frac{4 R^{1+2\sigma}}{3(1-|p|)(1+2\sigma)}.
\end{align*}
When $1/2 \leq |p|\leq 1$, for any $t\in [-1,1]$, we have the following estimate:
\begin{align*}
\frac{1-t^2}{-|p| t + \rho/2 +1} &\leq \frac{2(1-t)}{-|p|t + \rho/2 +1} \\
  &= \frac{2}{|p|}-\frac{2}{|p|}\cdot\frac{\rho/2 +1-|p|}{\rho/2 +1-|p|t} \\
  & \leq \frac2{|p|}. 
\end{align*}
Therefore, \red{we conclude}
\begin{align*}
D_{12}\left(p,\frac{p^2}{2}+\gamma_0\right) &\leq \int_0^R \int_{-1}^1 \frac{2\rho^{2\sigma} dtd\rho}{|p|} \\
&= \frac{4 }{1+2\sigma}\cdot \frac{R^{1+2\sigma}}{|p|}.
\end{align*}
Finally, when $|p|>1$, we have $z_0(|p|)=|p|-1/2+\gamma_0$ and
\begin{align*}
D_{12}\left(p,|p| - \frac12 + \gamma_0\right) &= \int_0^R\int_{-1}^1 \frac{(1-t^2)dt\rho^{1+2\sigma} d\rho}{(1/2)(\rho-|p|+1)^2+|p|\rho(1-t)} \\
    &\leq \int_0^R\int_{-1}^1 \frac{(1-t^2)dt\rho^{2\sigma} d\rho}{|p|(1-t)} \\
    &\leq \frac{4 }{1+2\sigma}\cdot \frac{R^{1+2\sigma}}{|p|}.
\end{align*}
Combining the \red{two preceding} estimates, we complete the proof.
\end{proof}

Next, we investigate the real solutions of Equation~(\ref{KKK}) in the region $\Omega^-$, to locate the eigenvalues (if any) of the operator $H(p)$.

\begin{lem}\label{dd2}
Let $(|p|,z)\in \Omega^-$ with $z$ real. 
The equation $F(p,z)=0$ has a unique solution $z=z^\ast(p)$ when $|p|\leq 1$ and there exists a constant $p_0>1$ such that it has no solution when $|p|>p_0$.
\red{This} constant $p_0$ \red{satisfies the estimate} 
\[
p_0 < 4+\frac{2\pi e^2 R^{1+2\sigma}}{1+2\sigma}.
\] 
Moreover, when $0\leq |p|\leq p_0$, the range of $z$ \red{has the following two-sided bound}:
\begin{equation}\label{Eig1}
   0< z\leq\frac{7}{2}+\frac{2\pi e^2R^{1+2\sigma}}{1+2\sigma}+\gamma_0.
\end{equation}
\end{lem}

\begin{proof}
Solving the equation $K(p,z)=0$ in the region $\Omega^-$ is equivalent to solving the equation 
\begin{equation}\label{KKK1}
F(p,z)=\frac12p^2-z+\gamma_0 - \pi e^2 D_{12}(p,z)\left(\frac12p^2+z-\gamma_0\right)=0.
\end{equation}
Note that
\begin{multline*}
F'_z(p,z) = -1-\pi e^2 D_{12}(p,z)\\
- \pi e^2\left(\frac12p^2+z-\gamma_0\right) \int_0^R \int_{-1}^1 \frac{(1-t^2)\rho^{1+2\sigma} dtd\rho}{( p^2/2-|p|\rho t + \rho^2/2+\rho+\gamma_0-z)^2}<0
\end{multline*}
and therefore, $F(p,z)$ is a decreasing function with respect to $z$ in the region $\Omega^-$, if $p^2/2+z-\gamma_0\geq 0$.
It is also easy to check that $F(p,z)>0$ if $z<\gamma_0 - p^2/2$.
For $0<|p|\leq 1$ and $z=z_0(|p|)$, we have
$$
F(p,z)=- \pi e^2 p^2D_{12}\left(p,\frac12 p^2+\gamma_0\right)< 0, \quad 
$$ 
and therefore, (\ref{KKK1}) \red{has a solution} at least when $0<|p|\leq 1$.

The upper bound in (\ref{Eig1}) follows from (\ref{eqz8}). Indeed, we have
\begin{align*}
F(p,z_0(|p|)) &= \frac12(|p|-1)^2-\frac{\pi e^2}{2} (p^2+2|p|-1)D_{12} \left( p,|p| - \frac{1}{2} + \gamma_0 \right) \\
&\geq \frac12(|p|-1)^2- (|p|+2) \frac{2\pi e^2 R^{1+2\sigma}}{1+2\sigma},
\end{align*}
and the latter expression is strictly positive if 
\[
|p|\geq 4+\frac{ 4 \pi e^2 R^{1+2\sigma}}{1+2\sigma}.
\]
It can be written in terms of $z$ as 
$$
z=|p|-1/2+\gamma_0\geq\frac72+\frac{4\pi e^2 R^{1+2\sigma}}{1+2\sigma}+\gamma_0.
$$
We prove the lower bound of (\ref{Eig1}) by  contradiction. 
Assume that there exists some $z\leq 0$ satisfying (\ref{KKK1}). 
Let $\mu=\pi e^2$ and rewrite Equation~(\ref{KKK1}) as 
\begin{equation}\label{k1}
 2(D_{12}d)\mu^2-((p^2+2z)D_{12}-2d)\mu+p^2-2z=0, 
\end{equation}
where $d=R^{2+2\sigma}/(1+\sigma)$. 
Since $\mu>0$, only the positive solutions of the equation are of interest, 
and therefore, the following system of inequalities should hold:
\begin{equation}
\begin{cases}
((p^2+2z)D_{12}-2d)^2-8D_{12}d(p^2-2z) \geq 0, \\
(p^2+2z)D_{12}-2d \geq 0,
\end{cases}
\end{equation}
which is equivalent to
$$
\begin{cases}
z'\geq -1-s+2\sqrt{2s}, \\
z'\geq s-1,
\end{cases}
$$  
when $z'\leq 0$.
Here,
\[
s=\frac{2d}{p^2D_{12}(p,z)} \quad \mbox{\rm and} \quad z'=\frac{2z}{p^2}.
\]
From the assumption $z'\leq0$, it follows that $s\leq 1$, which is equivalent to 
$$
D_{12}(p,z)\geq \dfrac{2d}{p^2}.
$$
The latter inequality is {\em not} true for any values of $e$ and $R$. 
Indeed, note that
\begin{align*}
D_{12}(p,z)-\dfrac{2d}{p^2}
   &= \int_0^R\int_{-1}^1 \frac{(1-t^2)dt\rho^{1+2\sigma} d\rho}{p^2/2-|p|\rho t+\rho^2/2+\rho +\gamma_0-z} \\
   &\qquad - \int_0^R\int_{-1}^1 \frac{dt\rho^{1+2\sigma} d\rho}{p^2/2}\\
   &= \int_0^R \int_{-1}^1 \frac{2(-(|p| t-\rho)^2/2 - \rho-\gamma_0+z)}{p^2( p^2/2-|p|\rho t + \rho^2/2+\rho +\gamma_0-z)} \,dt\rho^{1+2\sigma} d\rho < 0.
\end{align*}
\red{Therefore,} the equation $K(p,z)=0$ \red{does not admit any} solution in the region $\Omega^-$ when $z\leq 0$. 
\end{proof}
\section{Proof of the main results}\label{resseq5}
Summarizing the results from the previous sections, we \red{now} prove \red{our} main theorems.

\paragraph{Proof of Theorem~\ref{th1}.} 
Repeating the argument at the beginning of Section~\ref{resseq4}, we prove that the essential spectrum of $H(p)$ is $[z_0(|p|),+\infty)$ and the zeros of the function $K(p,z)$ are \red{the} only possible addition to the spectrum in the interval $(-\infty,z_0(|p|))$.

By Lemma~\ref{dd2}, for each $p\in \mathbb{R}^3$, there exists a solution to $K(p,z)=0$ and the range of these values of $z$ belongs to the interval 
\[
\left[ 0,\frac72+\frac{2\pi e^2 R^{1+2\sigma}}{1+2\sigma}+\gamma_0 \right).
\]
These solutions are the eigenvalues of $H(p)$ for each $p$, with eigenfunctions 
$$
\psi(p,k)=\begin{pmatrix}
        \dfrac{p^2}{T(p,z^*(p))} \\[3ex]
        -\dfrac{(p\cdot |G(k,1)\rangle )}{L(p,z^*(p))} \left(\dfrac{p^2}{T(p,z^*(p))}+1 \right) \\[3ex]
    -\dfrac{(p\cdot |G(k,2)\rangle )}{L(p,z^*(p))} \left(\dfrac{p^2}{T(p,z^*(p))}+1 \right)\\
    \end{pmatrix}.
$$
\red{This establishes Theorem~\ref{th1}.}
\qed

\paragraph{Proof of Theorem~\ref{th2}.} 
It is easy to show that  $E_{\sigma}(|p|)$ is \red{a} $C^2$-function. Moreover in Lemma~\ref{dd2}, we proved the equation $F(p,z)=0$ has a unique solution $z=z^*(p)$ for small $p$, which is equal to $E_{\sigma}(|p|)$. Using that $E_\sigma(0)=\gamma_0$, $E_\sigma'(0)=0$, and formula (\ref{eff1}), we \red{get} the desired result:
$$
\frac{\partial^2 E_{\sigma}(p)}{\partial |p|^2}\Bigg|_{p=0}
=-\frac{F''_{|p|}}{F'_z}\Bigg|_{p=0}
=\frac{1-\pi e^2D_{12}(0,\gamma_0)}{1+\pi e^2D_{12}(0,\gamma_0)}.
$$ 
\red{This establishes Theorem~\ref{th2}.}
\qed

\section*{Acknowledgments} The authors express their gratitude to the reviewers for helpful comments and valuable suggestions that greatly improved the manuscript.
 This work was supported by research Grant No.~12 of the Higher Education Reform Project, Mongolia.

\parbox{7.5cm}{\it \small Dagva Dayantsolmon (D. Dayantsolmon) \\
Department of Mathematics\\
National University of Mongolia \\
University Street 3 \\
Ulaanbaatar, Mongolia\\}

\parbox{6.5cm}{\it \small Artbazar Galtbayar (A. Galtbayar) \\
Department of Applied Mathematics\\
National University of Mongolia \\
University Street 3 \\
Ulaanbaatar, Mongolia} 

\end{document}